\documentclass[9pt,shortpaper,twoside,web]{ieeecolor}
\usepackage{generic}
\usepackage{cite}
\usepackage{amsmath,amssymb,amsfonts}
\usepackage{algorithmic}
\usepackage{graphicx}
\usepackage{textcomp}
\usepackage{bm}
\usepackage[ruled,vlined]{algorithm2e}

\newtheorem{theorem}{Theorem}[section]

\newtheorem{lemma}[theorem]{Lemma}
\newtheorem{assumption}{Assumption}

\def\BibTeX{{\rm B\kern-.05em{\sc i\kern-.025em b}\kern-.08em
    T\kern-.1667em\lower.7ex\hbox{E}\kern-.125emX}}
\begin{document}
\title{Episodically adapted network-based controllers}
\author{Sruti Mallik, \IEEEmembership{IEEE Student Member}, and ShiNung Ching  \IEEEmembership{Member, IEEE}
\thanks{Insert date of submission. S. Ching holds a Career Award at the Scientific Interface from the Burroughs-Wellcome Fund.This work was partially supported by grants 1653589 and 1724218 from the US National Science Foundation }
\thanks{All authors are with the Department of Electrical and Systems Engineering, Washington University, St. Louis, MO 63130, USA (e-mail: sruti.mallik@wustl.edu)}
\thanks{S. Ching is with the Department of Biomedical Engineering, Washington University, St. Louis, MO 63130, USA (e-mail: shinung@wustl.edu).}}

\maketitle

\begin{abstract}
We consider the problem of distributing a control policy across a network of interconnected units. Distributing controllers in this way has a number of potential advantages, especially in terms of robustness, as the failure of a single unit can be compensated by the activity of others. However, it is not obvious \textit{a priori} how such network-based controllers should be constructed for any given system and control objective. Here, we propose a synthesis procedure for obtaining dynamical networks that enact well-defined control policies in a model-free manner. We specifically consider an augmented state space consisting of both the plant state and the network states. Solution of an optimization problem in this augmented state space produces a desired objective and specification of the network dynamics.  Because of the analytical tractability of this method, we are able to provide convergence and robustness assessments.
\end{abstract}


\section{Introduction}
\label{sec:introduction}

Realizing network-based controllers is a long-running thread in systems control theory. The overarching question at hand is how one can distribute a given control policy across a population or network of constituent units whose collective action imparts the desired objective \cite{langbort2010distributed, gama2021graph}. One of the advantages of such an approach is in terms of robustness, since in principle, such a scheme allows for some units to fail while other units compensate, thus leaving the overall control policy intact.

Like conventional controllers, distributed controllers can be designed in model-based and model-free paradigms. The latter, relying on mechanisms of adaptation and learning \cite{ sutton1992reinforcement, aswani2013provably,  jiang2020learning} are particularly germane as control applications become increasingly centered on situations in which plant dynamics are variable or unknown \textit{a priori}. In this context, of particular significance are adaptive control\cite{aastrom2013adaptive} or iterative learning control algorithms \cite{bristow2006survey} which has received much attention for their ability to synthesize control under uncertainty. However, these algorithms do not necessarily focus on creating control or actuation signals through distributed computations.

One of the intersection points for the above issues is, of course, the domain of artificial neural networks (ANNs) and learning/optimization.  Such constructs are nominally capable of implementing a diversity of control laws and, potentially, learning these policies in an online, adaptive manner\cite{barto1983neuronlike, sontag1993neural, narendra1997adaptive, nguyen1990neural, sanchez2011control, lewis2020neural}. At a high level, ANNs are motivated by their biological counterparts, where we know that the aforementioned premise of robustness is fully enacted.  

There is intense current attention on ANNs in a variety of learning and adaption problems. However, there remain many open challenges and caveats in how to design these networks -- and, especially, recurrent ANNs (RNNs) -- for continuous state and input space control problems. For example, gradient descent and other first order learning methods for RNNs such as the ubiquitous backpropogation though time, struggle to retain long-term temporal dependencies \cite{bengio1994learning, silva1997bridging}, a crucial issue in the face of control problems where the plant dynamics may span several time-scales. Further, such methods typically require secondary assumptions regarding low rank network structure \cite{schuessler2020dynamics} or careful regularization of the objective \cite{martens2011learning} in order to ensure convergence to a viable control policy. Learning methods based on reward reinforcement \cite{song2017reward} are also possible, but likewise face issues of fragility and poorly understood convergence properties \cite{bertsekas1996neuro, bucsoniu2012least, van2012reinforcement, friedrich2019least}.

The overall goal of this paper is to advance the above considerations by proposing and studying a distributed, network-based control strategy together with an analytically tractable adaptation method, in order to control unknown dynamical systems.  In our previous work \cite{mallik2020neural, mallik2021topdown}, we developed model-based frameworks for synthesizing network dynamics that implement certain classical control policies. Our goals there were different than in the current work, as we were focused primarily on the emergent dynamics associated with objective functions intended to serve neuroscience endpoints. As such, we assumed perfect knowledge of plant dynamics.  Here, we leverage the potential for using these strategies to develop distributed controllers for physical systems, wherein plant dynamics may be unknown. 

 Thus, we engage the problem of model-free distributed control design. A high-level schematic of the problem at hand is depicted in Figure \ref{schematic_decoding}. Here, the activity of a network of units is responsible for controlling a dynamical systems, whose dynamics is unknown. The two operative questions we seek to answer are: (i) what should be the dynamics of units in the network and their interaction, and (ii) how can they be learned online. Importantly, we do not solve a sequential problem of system identification followed by control design \cite{moerland2020model, doya2002multiple, kaiser2019model}.  Instead, we attempt to build/learn our network-based controllers directly online without prior model inference. 
 To facilitate this tractably, our key development is the formulation of an augmented state space consisting of both the plant and network states. The latter is endowed with a generic integrator dynamics, which allows us to pose a single optimization problem for a control objective, whose solution determines the interconnectivity between network units. When this control objective is quadratic, we are able to deploy episodic adaptation to solve this problem in a model-free manner and, further, ascertain certain analytical convergence properties. 
 Further, we show that one of the benefits of distributing the solution in a network is the robustness in performance it provides in the event of neuronal failure \cite{kalampokis2003robustness, huang2019spiking}. For instance, when we remove contribution of a subset of network units before, during or after learning, the network adjusts and can still perform the task at hand.     

Therefore, the key contributions of this paper are: (a) synthesis of a dynamical network that can optimally control a lower dimensional physical system in a model-free manner through distributed network activity;(b) theoretical characterization of the conditions under which this iterative algorithm approaches the optimal policy; (c) analysis of the network performance to understand the properties of robustness under neuronal failure and the influence of different hyperparameters of the algorithm on the solution. 

The remainder of this paper are organized as follows. In Section 2, we formally introduce the mathematical details of the problem and approaches used to arrive at a solution. In Section 3, we present the key results of this research by demonstrating our framework through two numerical examples. Finally, in Section 4 we discuss the strengths of the proposed framework and outline topics that can be addressed through future research.

\begin{figure}[t]
	\centering
	\includegraphics[width = 0.90\linewidth]{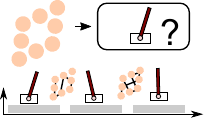}
	\caption{We consider the problem of constructing and parameterizing distributed, network-based controller for unknown systems. A base network architecture is analytically developed and then adapted over successive learning episodes.}
	\label{schematic_decoding}
\end{figure}

\section{Problem Formulation}

\subsection{Model Class}

We focus on a class of second order control problems.  Consider the linear plant dynamics:
\begin{equation}
	\label{pos}
	\bm{\Psi}_{t+1} = \bm{\Psi}_t + C\bm{\nu}_t
\end{equation}
\begin{equation}
	\label{vel}
	\bm{\nu}_{t+1} = A_{\Psi}\bm{\Psi}_t +  A_{\nu}\bm{\nu}_t + B_x\mathbf{x}_t
\end{equation}
where $\bm{\Psi}_t, \bm{\nu}_t \in \mathbf{R}^m$ are generalized position and velocities, respectively and $A_{\Psi}$, $A_{\nu}$ $C$ are matrices of appropriate dimension.

At the heart of our formulation is the vector $\mathbf{x}_t \in \mathbb{R}^n$, which contains the activity of $n$ units over which we seek to distribute the control of the system at hand. The activity of these $n$ units is linearly projected onto the velocity dynamics via $B_x \in \mathbb{R}^{m \times n}$. A key premise is that the dimensionality of the network-based controller is much larger than that of the plant, i.e., $n > m$.  Then, \textit{the overall problem is how to specify the time evolution of $\mathbf{x}_t$ i.e., network dynamics so as to meet a desired control objective, when  $A_{\Psi}$, $A_{\nu}$, $B_x$, $C$ are unknown.}


To probe this question, we consider the following high-level objective function function: 
\begin{equation}
	\label{first_opt}
	\mathbb{J} = \frac{1}{2}\sum_{t = 0}^{\infty}[q_J(\bm{\Psi}_t, \bm{\nu}_t) + \mathbf{x}_t^T\mathbf{S}\mathbf{x}_t + \mathbf{u}_t^T\mathbf{R}\mathbf{u}_t]
\end{equation}
wherein we introduce the secondary dynamics
\begin{equation}
	\mathbf{x}_{t+1} = \mathbf{x}_t + \Delta t \mathbf{u}_t,
\end{equation}
again emphasizing that $\mathbf{x}_t$ here is not the state of the plant, but rather of the network that will be enacting control of the plant. This objective function can be interpreted as follows. The first term here quantifies the quality of performance of the system in the low-dimensional space, the second term quantifies the energy expenditure of the network, while the third term acts as a regularizer preventing arbitrary large changes in network activity. 

If the dynamics of the model (i.e., $A_{\Psi}$, $A_{\nu}$, $B_x$, $C$) were known and if $q_J$ is specified as a quadratic function, then this problem reduces to a classical optimal control problem, i.e., the infinite horizon Linear Quadratic Regulator \cite{anderson2007optimal}. When the complete dynamics of the model is unknown, then a solution to this optimization problem can be found by either prior model system identification \cite{moerland2020model, kaiser2019model}, or through `model-free' adaptive control design frameworks \cite{bradtke1994adaptive, degris2012model}.  Our goal is to specify the $\mathbf{x}_t$ dynamics without prior system identification (i.e., to learn/construct these dynamics online).


\subsection{Network synthesis problem}

To realize the $\mathbf{x}_t$ dynamics, we transform the problem by defining, $\bm{\Omega}_t \equiv [\bm{\Psi}_t^T, \bm{\nu}_t^T, \mathbf{x}_t^T]^T \in \mathbb{R}^{2m + n}$ . Combining the dynamics from \eqref{pos} and \eqref{vel}, we have:
\begin{equation}
	\label{discrete_sys}
	\bm{\Omega}_{t+1} = \mathbf{A}\bm{\Omega}_{t} + \mathbf{B}\mathbf{u}_t 
\end{equation}
Here, 
\begin{equation}
\mathbf{A} = \left[\begin{array}{ccc}
	\mathbf{I}_m &C &\mathbf{0}\\
	A_{\Psi} & A_{\nu} &B_x \\
	\mathbf{0} &\mathbf{0} &\mathbf{I}_n
\end{array}\right], \mathbf{B} = \left[\begin{array}{c}
\mathbf{0} \\
\mathbf{0} \\
\Delta t \mathbf{I}_n
\end{array}\right]
\end{equation}
and $\Delta t$ is the sampling interval and we assert that the initial state is selected from a fixed distribution i.e., $\bm{\Omega}_0 \sim \mathcal{D}$. With the definition of $\mathbf{\Omega}_t$, we can now write the objective function as:
\begin{equation}
	\label{first_opt_discrete}
	\begin{aligned}
		\mathbb{J}_d &=\frac{1}{2}\sum_{0}^{\infty}[\bm{\Omega_{t}}^T\mathbf{Q}\bm{\Omega_{t}}+\mathbf{u}_t^T\mathbf{R}\mathbf{u}_t]		
	\end{aligned}
\end{equation}
 The goal here is to synthesize the dynamics of the network i.e., $\mathbf{u}_t$ that optimizes this objective function without explicit knowledge of model dynamics. In other words, we need to solve the following optimization problem without knowing $A$ and $B$. 

\begin{equation}
	\begin{split}
	\label{op_problem}
	\underset{\mathbf{u}_t}{argmin}\hspace{0.5mm}\frac{1}{2}\sum_{0}^{\infty}[\bm{\Omega_{t}}^T\mathbf{Q}\bm{\Omega_{t}}+\mathbf{u}_t^T\mathbf{R}\mathbf{u}_t]\\
	\text{subject to} \hspace{0.5mm} 	\bm{\Omega}_{t+1} = \mathbf{A}\bm{\Omega}_{t} + \mathbf{B}\mathbf{u}_t 
	\end{split}
\end{equation}
For ease of notation, we will use $\mathbf{c}_t \equiv [\bm{\Omega_{t}}^T\mathbf{Q}\bm{\Omega_{t}}+\mathbf{u}_t^T\mathbf{R}\mathbf{u}_t]$ in subsequent sections. Note that while we consider a discrete time problem, the state space $\bm{\Omega} = \mathbb{R}^{2m+n}$ and the action space $\mathbf{U} = \mathbb{R}^n$ are infinite and continuous. 

\section{Results}

\subsection{Episodically adapting distributed network as controller}

\subsubsection{Online Least Squares Approximate Policy Iteration}

We define the activity of the network as policy, i.e., $\mathbf{u}_t = \pi(\bm{\Omega}_t)$. Starting from the state $\bm{\Omega}_t$ and following a policy $\pi$, the cost-to-go or the value function \cite{sutton2018reinforcement, lewis2012reinforcement} is given by:
\begin{equation}
	\label{cost_to_go}
	\mathrm{V}_{\pi}(\bm{\Omega}_t) = \frac{1}{2} \sum_{t}^{\infty} \mathbf{c}_t
\end{equation}
We introduce a discount factor $0 \leq \gamma \leq 1$ to \eqref{cost_to_go}, in order to bias the cost to immediate errors:
\begin{equation}
	\label{cost_to_go_discounted}
	\mathrm{V}_{\pi}(\bm{\Omega}_t) = \frac{1}{2} \sum_{i = 0}^{\infty} \gamma^i \mathbf{c}_{t + i}
\end{equation}
Using the definition in \eqref{cost_to_go_discounted}, we can specify
a state-action value function $\mathrm{Q}_{\pi}$ \cite{bradtke1994adaptive,watkins1989learning}
\begin{equation}
	\label{Q_val}
	\mathrm{Q}_{\pi}(\Omega, u) = \mathbf{c}(\Omega, u) + \gamma \mathrm{V}_{\pi}(\bm{\Omega}_{t+1}),
\end{equation}
 the sum of the one step cost incurred as a result of taking action $u$ from state $\Omega$ and following the policy $\pi$ from there onward. Now, we can write $\mathrm{V}_{\pi}(\Omega) = \mathrm{Q}_{\pi}(\Omega, \pi(\Omega))$ and, thus
\begin{equation}
	\label{Q_val_modified}
	\mathrm{Q}_{\pi}(\bm{\Omega}_t, \mathbf{u}_t) = \mathbf{c}(\bm{\Omega}_t, \mathbf{u}_t) + \gamma \mathrm{Q}_{\pi}(\bm{\Omega}_{t+1}, \pi(\bm{\Omega}_{t+1})).
\end{equation}
Therefore, the state-action value function $\mathrm{Q}_{\pi}$ can be computed using:
\begin{equation}
	\label{quadratic}
	\mathrm{Q}_{\pi}(\bm{\Omega}_t, \mathbf{u}_t) = \left[\begin{array}{cc}
		\bm{\Omega}_t &\mathbf{u}_t \end{array}\right] \mathbf{H}_{\pi} \left[\begin{array}{c}
	\bm{\Omega}_t \\
	\mathbf{u}_t 
\end{array}\right]
\end{equation}
If the model dynamics were known, then the optimal strategy could be derived directly by taking the derivative of the right hand side of \eqref{Q_val_modified} and setting it to zero, i.e.,
\begin{equation}
	\label{optimal_known_dynamics}
	\begin{aligned}
		\pi(\bm{\Omega}_t) &= \underset{\mathbf{u}_t}{\text{argmin}}\hspace{1mm}\mathbf{c}(\bm{\Omega}_t, \mathbf{u}_t) + \gamma \mathrm{Q}_{\pi}(\bm{\Omega}_{t+1}, \pi(\bm{\Omega}_{t+1})) \\
		&= -\gamma(\mathbf{R}+\gamma \mathbf{B}'\mathbf{P}_{\pi}\mathbf{B})^{-1}\mathbf{B}'\mathbf{P}_{\pi}\mathbf{A}\bm{\Omega}_t \\
		&= -\mathbf{H}_{\pi(22)}^{-1} \mathbf{H}_{\pi(21)}\bm{\Omega}_t,
	\end{aligned}
\end{equation}
$\mathbf{H}_{\pi} \in \mathbb{R}^{(n'+n) \times (n'+n)}$ is:
\begin{equation}
	\label{Q_val_rhs}
	\mathbf{H}_{\pi} = 
	\left[\begin{array}{cc}
		\mathbf{Q} + \gamma \mathbf{A}'\mathbf{P}_{\pi}\mathbf{A} &\gamma \mathbf{A}'\mathbf{P}_{\pi}\mathbf{B} \\
		\gamma \mathbf{B}'\mathbf{P}_{\pi}\mathbf{A} &\mathbf{R}+\gamma \mathbf{B}'\mathbf{P}_{\pi}\mathbf{B}
	\end{array}\right]
\end{equation}
and $n' = 2m+n$. The details of the derivation for \eqref{optimal_known_dynamics} can be found in \cite{bradtke1994adaptive}. However, in our formulation the plant dynamics are unknown, necessitating formulating the state action value function in a data-driven fashion. 

The basic idea of our adaptation approach is based on the quadratic nature of the state action value function, which can therefore be written as: 
\begin{equation}
	\label{approximate}
	\mathrm{Q}_{\pi}(\bm{\Omega}_t, \mathbf{u}_t) = \bm{\Theta}^T \phi_t
\end{equation}
Here, $\bm{\Theta}  = vec(\mathbf{H}_{\pi})$ and 
\[
\phi_t = \left[\begin{array}{c}
	\bm{\Omega}_t \\
	\mathbf{u}_t
\end{array}\right] \otimes \left[\begin{array}{c}
\bm{\Omega}_t \\
\mathbf{u}_t
\end{array}\right].
\] 
If we can obtain an estimate $\hat{\bm{\Theta}}$ of the true parameters $\bm{\Theta}$, then we can likewise estimate the state action value function. 
We start by assuming a linear policy:
\begin{equation}
	\label{linear_policy}
	\pi(\mathbf{\Omega}_t) \equiv \mathbf{W}\mathbf{\Omega}_t
\end{equation}
Substituting \eqref{linear_policy} in the state action value function for a chosen policy $\pi$ in \eqref{Q_val_modified}, yields the following error:
\begin{equation}
	\label{error}
	\mathbf{e}_t = \mathbf{c}_t - \hat{\bm{\Theta}}^T\Phi_t
\end{equation}
where,
\[\Phi_t = \left[\begin{array}{c}
	\bm{\Omega}_t \\
	\mathbf{u}_t
\end{array}\right] \otimes \left[\begin{array}{c}
	\bm{\Omega}_t \\
	\mathbf{u}_t
\end{array}\right] - \gamma \left[\begin{array}{c}
\bm{\Omega}_{t+1} \\
\mathbf{W}\bm{\Omega}_{t+1}
\end{array}\right] \otimes \left[\begin{array}{c}
\bm{\Omega}_{t+1} \\
\mathbf{W}\bm{\Omega}_{t+1}
\end{array}\right].\]
Therefore, we estimate the elements of the matrix $\mathbf{H}_{\pi}$ as
\begin{equation}
	\label{error_min}
	\hat{\bm{\Theta}} = \underset{}{\text{argmin} \hspace{2mm}} \sum_{t = 1}^{T} \mathbf{e}_t^2
\end{equation}
Here, $T$ is the horizon for which the data is collected by probing the system. Combining equations \eqref{error_min} and \eqref{optimal_known_dynamics} leads us to Algorithm \ref{algo}, an \textit{approximate online least squares policy iteration}. In essence, we begin with an initial choice of a policy $\pi_k$ and use it to collect samples for an episode comprising of $T$ timesteps. Next, we use the data collected to estimate $\mathrm{Q}_{\pi}$. Based on this estimate, we compute an updated policy $\pi_{k+1}$ with which we probe the system next. We continue to repeat these steps till the policy converges (see Algorithm \ref{algo}). Here, 
$\mathbf{C}_k = \left[\mathbf{c}_{1|k}, ..., \mathbf{c}_{T|k} \right] \in \mathbb{R}^{1 \times T}$ and $\bm{\Phi}_k = \left[\Phi_{1|k},... \Phi_{T|k}\right] \in \mathbb{R}^{L \times T}$ where $L = (n' + n)^2$ is the number of parameters to be estimated. We incorporate exploratory action in the input signal at every iteration $k$, wherein the input that is used to collect data is given as:
\begin{equation}
	\label{exploration}
	\mathbf{u}_t \equiv \pi_k(\bm{\Omega_{t}}) + \varepsilon_t = \mathbf{W}_k\mathbf{\Omega}_t +  \varepsilon_t
\end{equation}
where, $\varepsilon_t \sim \mathcal{N}(0, \sigma_y^2)$. This exploration technique injects randomization while ensuring that a stabilizing policy in the neighborhood of the updated policy is considered for generating data during each episode of the algorithm, safeguarding against badly scaled solutions. 

 There are several key distinctions between the above framework and least square policy iteration (LSPI) \cite{lagoudakis2003least, bucsoniu2010online}. LSPI operates off-line by performing policy improvements only when the state action value function has been estimated accurately. In contrast, in Algorithm \ref{algo}, the policy is updated on every episode.
\begin{algorithm}[t]
	\SetAlgoLined
	k = 0, Initial policy $\pi_0$, Initial state $\bm{\Omega}_0 \sim \mathcal{D}$\;
	\Repeat{$||\hat{\mathbf{H}}_{\pi|k+1} - \hat{\mathbf{H}}_{\pi|k}||<Tol$}{
		t = 0\; $D = \left\{ \right\}$\;
		\For{t = 1,..., T}{
			$\mathbf{u}_t = \pi_k(\bm{\Omega}_t) + \varepsilon_t$\;
			Apply $\mathbf{u}_t$ and measure the next state $\bm{\Omega}_{t+1}$ and the cost $\mathbf{c}_t$\;
			$D = D \bigcup \left\{ \bm{\Omega}_{t}, \mathbf{u}_t, \mathbf{c}_t, \bm{\Omega}_{t+1}\right\}$\;
		}
	\textbf{Policy Evaluation}\\
	$\hat{\bm{\Theta}}_k = (\bm{\Phi}_k^T)^{\dagger}\mathbf{C}_k^T$\;
	Unpack into matrix $\hat{\mathbf{H}}_{\pi|k}$\;
	\textbf{Policy Update} \\
	$\pi_{k+1}(\bm{\Omega}_{t}) = -\hat{\mathbf{H}}_{\pi(22)|k}^{-1} \hat{\mathbf{H}}_{\pi(21)|k}\bm{\Omega}_t$\;
	$k = k + 1$\;		
	}
	\caption{Online Least Squares Approximate Policy Iteration}
	\label{algo}
\end{algorithm}
Another key difference between this work and \cite{bucsoniu2010online} is how we construct the policy that controls the system. In \cite{bucsoniu2010online}, the online LSPI algorithm is used to construct a control signal directly. In contrast, here we derive the dynamics of a network that \textit{generates} the control signal. In other words, the optimal policy is distributed across the high-dimensional network and embedded in the dynamics of its units.  A second difference between this work and \cite{bucsoniu2010online} is in how we include exploration in our algorithm. Online LSPI must explore to ensure that it provides a good estimate of the state action value function. In \cite{bucsoniu2010online}, an $\epsilon$-greedy exploration is used: at every timestep, a random exploratory action is applied with a probability of $\epsilon_k \in [0,1]$.  In contrast, we inject randomization here in a manner such that policies only in the neighborhood of the current derived policy are considered in each episode. This ensures that an arbitrary random policy for which the system exhibits unbounded behavior is not considered during policy iteration. This is important as unbounded behavior of the system in any episode could negatively impact the convergence of the algorithm (we discuss convergence of the algorithm further in Section C, below). 

\subsubsection{Interpretation of the online least squares policy iteration as a two timescale dynamical network }

The algorithm for obtaining the optimal solution begins with an initial `guess' for a strategy that will enable the performance of the motor control task. Thereafter, through interactions with the environment the network evaluates the current policy and updates it. This process continues until the network learns a strategy that can accomplish the task at hand optimally. It must be noted that in this work, we are not estimating the dynamics of the model \textit{per se}. We are instead performing data-driven optimization that directly leads us to a network enacting the optimal policy. 
\begin{figure}[h]
	\centering
	\includegraphics[width = \linewidth]{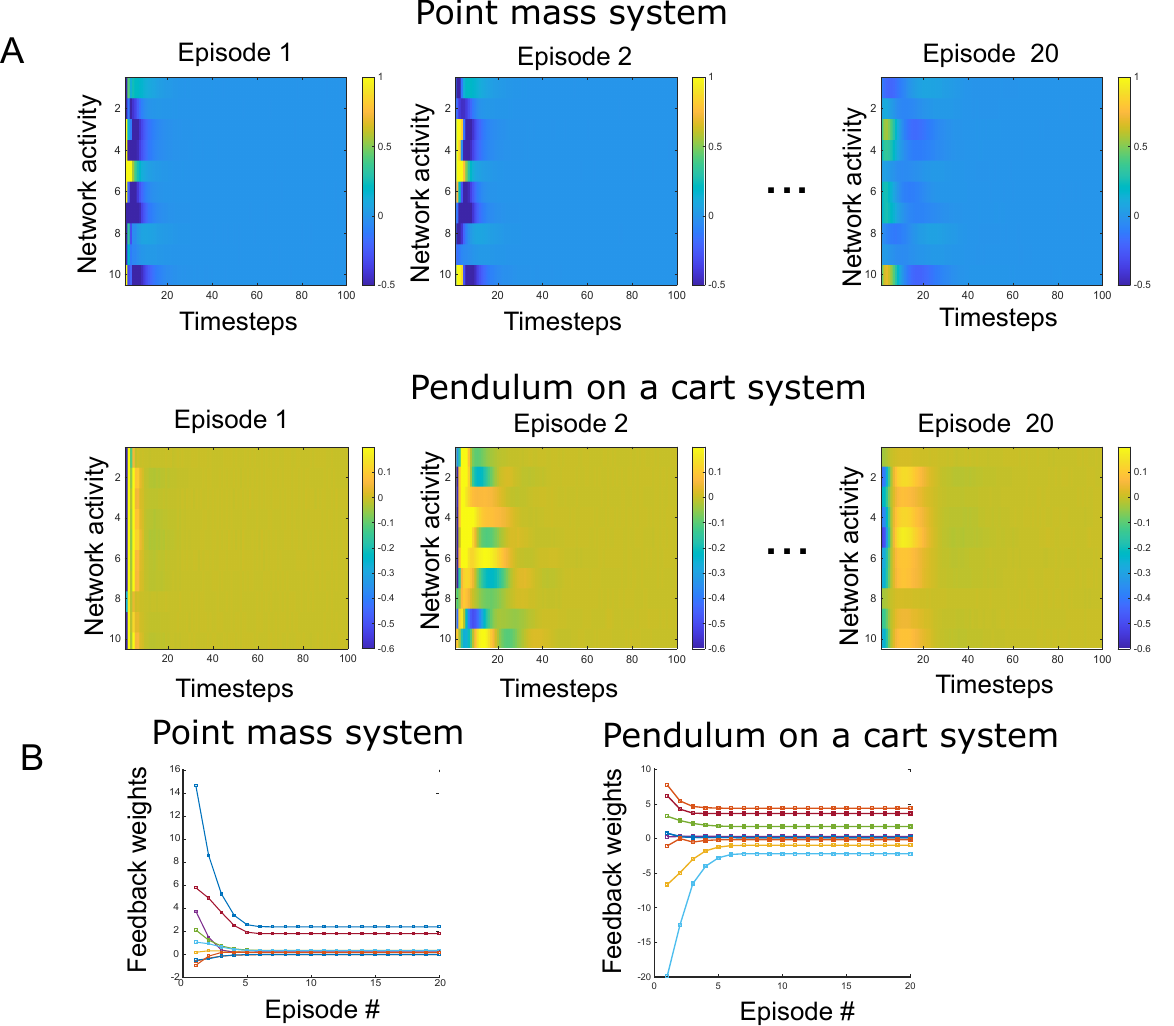}
	\caption{A. Network activity during the first 100 timesteps of episodes (top) point-mass system (bottom) inverted pendulum on a cart. (Note that activity here is represented as changes wrt a positive baseline). B. Network connections evolving over episodes (left) point-mass system (right) inverted pendulum on a cart. }
	\label{synaptic_weights}
\end{figure}
There are two entities in our network:(a) states associated with the network activity, and (b) the feedback matrix that undergoes adaptation to learn and perform the task at hand. 

We can identify two time-scales \cite{lewis2009reinforcement} in Algorithm \ref{algo}: first, a slow (mediated by $k$) timescale associated with adaption of the feedback weights:
\begin{equation}
	\label{outer_time_scale}
	\begin{aligned}
		\pi_{k+1}(\bm{\Omega}_{t}) &= \mathbf{W}_{k+1}\bm{\Omega}_t \\
		&\equiv -\hat{\mathbf{H}}(\bm{\Theta}_k)_{22}^{-1} \hat{\mathbf{H}}(\bm{\Theta}_k)_{21}\bm{\Omega}_t 		
	\end{aligned}
\end{equation}
and second, a faster timescale (mediated by $t$) associated with the dynamics of the network itself:
\begin{equation}
	\label{inner_time_scale}
	\begin{aligned}
		\mathbf{x}_{t+1} &= \mathbf{x}_t + \Delta t \mathbf{u}_t \\
		&\equiv \mathbf{x}_t + \Delta t \mathbf{W}_k \bm{\Omega}_t
	\end{aligned}
\end{equation}
The slow dynamics here (Equation \eqref{outer_time_scale}) directly arises from the episodic policy updates (see Algorithm 1). On the other hand, the network operates through the fast dynamics (Equation \eqref{inner_time_scale}) and actuates the physical system. 

If we write  $\mathbf{W}_k = [\mathbf{W}_k^{\Psi}: \mathbf{W}_k^{\nu}:\mathbf{W}_k^{\mathbf{x}}]$, then 
\begin{equation}
	\label{inner_time_scale_2}
	\begin{aligned}
		\mathbf{x}_{t+1} &= \mathbf{x}_{t}+\Delta t [\mathbf{W}_k^{\Psi} \bm{\Psi}_t + \mathbf{W}_k^{\nu} \bm{\nu}_t + \mathbf{W}_k^{\mathbf{x}}\mathbf{x}_t]\\
	\end{aligned}
\end{equation}
Equations \eqref{outer_time_scale} and \eqref{inner_time_scale_2} embody the two timescales. 


\begin{figure}[t]
	\centering
	\includegraphics[width=\linewidth]{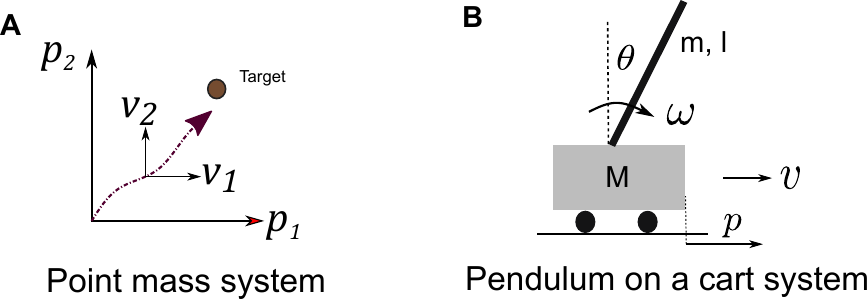}
	\caption{Schematic of tasks to be performed or systems to be controlled. }
\end{figure}

\subsection{Convergence analysis}

In Algorithm \ref{algo}, we have shown the iterative procedure that performs policy evaluation by estimating the state action value function episodically and thereafter updates the policy based on the current evaluation till convergence occurs. We further study conditions that must be satisfied to ensure that this algorithm converges to a policy that lies in proximity to the optimal policy. In \cite{bertsekas1996neuro, bucsoniu2012least, friedrich2019least}, the authors point out that as policy improvement step of this genre of algorithms inherently depend upon the estimate of the state action value function, it is difficult to ascertain the convergence properties. In \cite{ma2011approximate}, it has been shown that convergence analysis can be established in instances where approximate policy iteration is performed using basis function approximation.

In this work, we deal with a model-free optimization problem, where the cost at any time step assumes a quadratic form and the system transitions from one state to the next by linear dynamics whose specific parameters are not known. For this specification, we are able to make arguments regarding convergence. 

\begin{assumption}
	The estimate of the state action value function $\hat{\mathrm{Q}}_{\pi_k}$ determines the subsequent updated policy $\pi_{k+1}$. Therefore, we can define a Lipschitz continuous function $\Gamma_Q$ \cite{perkins2003convergent} such that:
	\begin{equation}
		\pi_{k+1}(\bm{\Omega}_t) = \Gamma_Q(\hat{\mathrm{Q}}_{\pi_k})
		\label{proof_assumption_1}
	\end{equation}
\end{assumption}

\begin{assumption}
	The least squares problem posed in Equation \eqref{error_min} is well-posed and the estimates of the parameters $\bm{\Theta}$ are bounded.
\end{assumption}

\begin{proof}
	We can posit a well-posed least squares problem by choosing episode length $T$ appropriately. We have discussed how to choose an appropriate episode length further in the next section.
\end{proof}

\begin{lemma}
	If $\pi_1$ and $\pi_2$ are two stabilizing policies then, there exists constants $\beta_{\phi}$, $\beta_{c}$ and $\beta_{\Phi}$ such that $||\phi_t^1 - \phi_t^2|| \leq \beta_{\phi} ||\pi_1 - \pi_2||$, $||\mathbf{c}_t^1 - \mathbf{c}_t^2|| \leq \beta_{c} ||\pi_1 - \pi_2||$ and $||\Phi_t^1 - \Phi_t^2|| \leq \beta_{\Phi} ||\pi_1 - \pi_2||$ for all $t$. Here, the superscripts indicate the policies under which the quantities are computed.
	\label{lemma_1}
\end{lemma}

\begin{proof}
	The proof for this Lemma follows from \cite{perkins2003convergent}.
\end{proof}

\begin{lemma}
	If $\pi_k$ is a sequence of stabilizing policies, then the sequence $\pi_k$ converges as $k \rightarrow \infty$ if Assumptions 1, 2 and Lemma \ref{lemma_1} holds.
	
	\label{lemma_2}
\end{lemma}

\begin{proof}
	Let us define a function $\xi_{k+1}(\bm{\Omega}_t) = ||\pi_{k+1}(\bm{\Omega}_t) - \pi_k(\bm{\Omega}_t)||$. Using \eqref{proof_assumption_1}, we can write:
	
	\begin{equation}
		\xi_{k+1}(\bm{\Omega}_t) \equiv ||\pi_{k+1} - \pi_k|| \leq \beta_Q ||\hat{
			\mathrm{Q}}_{\pi_k} - \hat{\mathrm{Q}}_{\pi_{k - 1}}||
		\label{Lipschitz_1}
	\end{equation}
	where $\beta_Q$ is the Lipschitz constant.
	
	Now, using \eqref{approximate} and Lemma \eqref{lemma_1}, we can further write: 
	\begin{equation}
		\begin{aligned}
			\xi_{k+1}(\bm{\Omega}_t) &\leq \beta_Q (\zeta_{\phi}||\hat{\bm{\Theta}}_k - \hat{\bm{\Theta}}_{k-1} || + \zeta_{\Theta}\beta_{\phi}||\pi_k - \pi_{k - 1}||)\\
			&= \beta_Q (\zeta_{\phi}||\hat{\bm{\Theta}}_k - \hat{\bm{\Theta}}_{k-1} || + \zeta_{\Theta}\beta_{\phi} \xi_k(\bm{\Omega}_t))
		\end{aligned}
		\label{proof_assumption_2}
	\end{equation}
	where, under stabilizing policies, we have $||\phi_t|| \leq \zeta_{\phi}$ and $||\hat{\bm{\Theta}}|| \leq \zeta_{\Theta}$. Then, it follows from Equations \eqref{error} and \eqref{error_min}: 
	\begin{equation}
		\begin{aligned}
			\bm{\Phi}_k^T \hat{\bm{\Theta}}_k - \bm{\Phi}_{k-1}^T \hat{\bm{\Theta}}_{k-1} &= \mathbf{C}_k^T - \mathbf{C}_{k-1}^T\\
			\bm{\Phi}_k^T (\hat{\bm{\Theta}}_k - \hat{\bm{\Theta}}_{k-1})+(\bm{\Phi}_k^T - \bm{\Phi}_{k-1}^T)\hat{\bm{\Theta}}_{k-1} &= \mathbf{C}_k^T - \mathbf{C}_{k-1}^T
		\end{aligned}
		\label{proof_assumption_3}
	\end{equation}
	For definition of $\bm{\Phi}_k$ and $\mathbf{C}_k$ refer to the Problem Formulation section. Combining Equation \eqref{proof_assumption_3} and Lemma \ref{lemma_1}, we can state:
	\begin{equation}
		\begin{aligned}
			||\hat{\bm{\Theta}}_k - \hat{\bm{\Theta}}_{k-1} || &\leq \frac{(\beta_c - \zeta_{\Theta}\beta_{\Phi})}{\zeta_{\Phi}}||\pi_k - \pi_{k - 1}|| \\
			&= \frac{(\beta_c - \zeta_{\Theta}\beta_{\Phi})}{\zeta_{\Phi}}\xi_k(\bm{\Omega}_t)
		\end{aligned}
		\label{proof_assumption_4}
	\end{equation}
	Here, $||\bm{\Phi}|| \leq \zeta_{\Phi}$. Using Equation \eqref{proof_assumption_4} in \eqref{proof_assumption_2}, we can write:
	\begin{equation}
		\begin{aligned}
			\xi_{k+1}(\bm{\Omega}_t) &\leq \beta_Q (\zeta_{\phi}\frac{(\beta_c - \zeta_{\Theta}\beta_{\Phi})}{\zeta_{\Phi}} + \zeta_{\Theta}\beta_{\phi}) \xi_k(\bm{\Omega}_t) \\
			&= \beta\xi_k(\bm{\Omega}_t)
		\end{aligned}
		\label{proof_final}
	\end{equation}
	
	Here, $\beta = \beta_Q (\zeta_{\phi}\frac{(\beta_c - \zeta_{\Theta}\beta_{\Phi})}{\zeta_{\Phi}} + \zeta_{\Theta}\beta_{\phi})$. If $0 <\beta<1$, then $\xi_k(\bm{\Omega}_t)$ goes to zero as $k \rightarrow \infty$. 
\end{proof}
At this point, we have characterized circumstances under which the policy iteration converges. Next, we show that the policy $\pi_k$ approaches the optimal policy $\pi_*$. This can be established following from the Lemma \cite{singh1994upper, kakade2003sample}:
\begin{lemma}
	Let $\mathrm{Q}_*$ is the optimal state-action value function and we have an estimate $\hat{\mathrm{Q}}$ of $\mathrm{Q}_*$ such that the deviation of $\hat{\mathrm{Q}}$ is bounded by $\epsilon$, i.e., $||\hat{\mathrm{Q}} - \mathrm{Q}_*|| \leq \epsilon$. Let $\pi$ be the policy wrt $\hat{\mathrm{Q}}$. Then for all states $\bm{\Omega}_t$, $||\mathrm{V}_* - \mathrm{V}_{\pi}||\leq \epsilon'$ and $\epsilon'$ depends on $\epsilon$. 
	\label{lemma_3}
\end{lemma}
\begin{proof}
	The proof of this Lemma follows from \cite{singh1994upper, kakade2003sample} and has been included in the Appendix of this paper for completeness.
\end{proof}

This indicates that if the policy iteration converges to any policy $\pi$, which is stabilizing and such that the deviation of the corresponding state action value function $\hat{\mathrm{Q}}$ from the optimal state action value function $\mathrm{Q}_*$ can be bounded by $\epsilon$, then our policy does not get any further away from the optimal policy by a factor of $\epsilon$.

On the basis of the above intermediate results, we are now ready to state the conditions under which our proposed algorithm converges.

\begin{theorem}
	If $\pi_k$ is a sequence of stabilizing policies, then $\pi_k$ converges to a policy in the neighborhood of the optimal policy $\pi_*$ as $k \rightarrow \infty$, provided there exists a bounded solution to the least squares problem at each episode.
\end{theorem}

\begin{figure}[t]
	\centering
	\includegraphics[width=\linewidth]{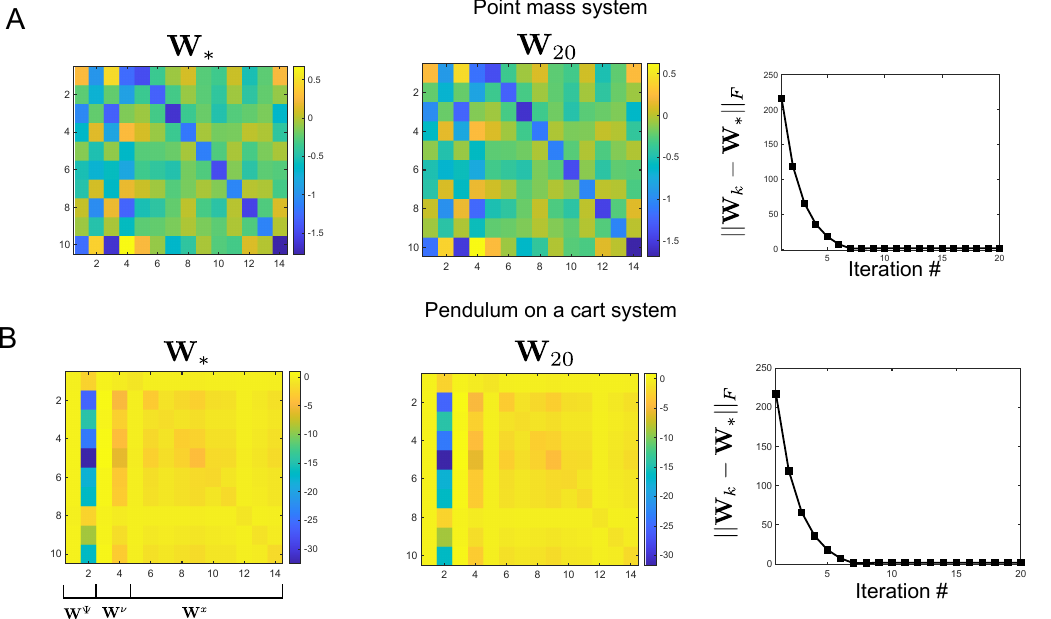}
	\caption{(Left) Feedback matrix for known model dynamics (Middle) Feedback matrix from model free policy iteration. (Right) Convergence behavior of the algorithm over iterations. A. Point mass system and B. Pendulum on a cart system. In the bottom left panel, we have shown the columns which correspond to the sub-matrices $\mathbf{W}_{\Psi}$, $\mathbf{W}_{\nu}$ and $\mathbf{W}_x$ of equation \eqref{inner_time_scale_2}.}
	\label{convergence_fig}
\end{figure}

Figure \ref{convergence_fig} shows how the algorithmic progresses to find the optimal feedback matrix in our two numerical example systems.

\subsection{Numerical examples}

Prior to convergence analysis, we explore two numerical examples to gain some intuition for the Algorithm: (i) spatial navigation of a unit point mass, and (b) stabilization of a pendulum in an inverted position (see Figure \ref{schematic_decoding}). In Figure \ref{synaptic_weights}A, B, we have shown an how the network activity evolves when it solves these tasks. For the ease of description here, we have shown in Figure \ref{synaptic_weights}B, how nine randomly chosen elements of the feedback matrix $\mathbf{W}_k$ adapts episodically. We now delve into these examples in more detail.


\subsubsection{Spatial navigation of a point mass}

In the first numerical example, we look at the spatial navigation of a unit point mass (see Figure \ref{point_mass_fig} A). Without any loss of generality we assume that the point mass moves in a two-dimensional plane i.e., $m = 2$. The motion of the point-mass is governed by the following equations:
\begin{equation}
	\label{point_mass}
	\mathbf{p}_{t+1} = \mathbf{p}_t +\Delta t \mathbf{v}
\end{equation}

and,
\begin{equation}
	\label{force_point_mass}
	\mathbf{v}_{t+1} = (1 - \Delta t\lambda_v )\mathbf{v}_t +  \Delta t \mathbf{b}\mathbf{x}_t
\end{equation}
Here, the variables $\mathbf{p}_t$ and $\mathbf{v}_t$ correspond to position and velocity, respectively, of the point mass. $0<\lambda_v<1$ captures possible dissipation due to friction during motion. The goal here is to drive the point mass to a fixed target location $\mathbf{p}_T$ in the plane starting from the origin of the plane in an energy efficient manner (see Equation \eqref{center_out_cost}), leading to the cost:
\begin{equation}
	\label{center_out_cost}
	\mathbb{J} = \frac{1}{2}\sum_{t = 0}^{\infty}[(\mathbf{p}_t - \mathbf{p}_T)^T\mathbf{Q}_1(\mathbf{p}_t - \mathbf{p}_T) + \mathbf{x}_t^T\mathbf{S}_1\mathbf{x}_t + \mathbf{u}_t^T\mathbf{R}_1\mathbf{u}_t]
\end{equation}
This control task is motivated by a standard experiment in neuroscience: the center-out reaching task \cite{so2012redundant, sanchez2011control}. We find that beginning from an initial guess, the network through adaptation of feedback weights quickly learns the optimal policy to navigate to $\mathbf{p}_T$ (see Figure \ref{point_mass_fig} B-C).
\begin{figure}[b]
	\centering
	\includegraphics[width = \linewidth]{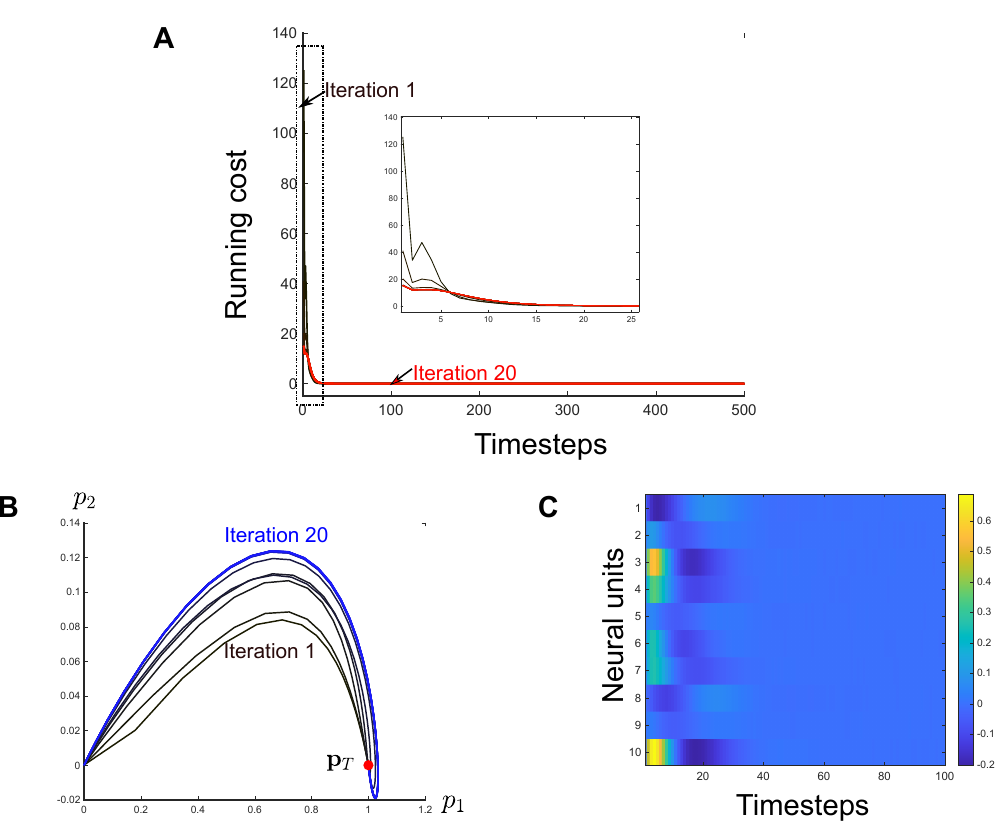}
	\caption{\textbf{Point mass system} A-B. The network quickly learns the optimal strategy to perform the task. (The running cost in the first 25 timesteps of each iteration is shown inset in A). C. Activity of the network after convergence to an optimal strategy. }
	\label{point_mass_fig}
\end{figure}
For the simulations we have $n = 10$, $\lambda_v = 0.25$, $\mathbf{p}_T = [1, 0]^T$, $\Delta t = 0.1$ and weights of $\mathbf{b} \in \mathbb{R}^{m \times n}$ chosen from a uniform random distribution. We additionally have $\sigma_y^2 = 0.01$, $\gamma = 0.99$, $\mathbf{Q}_1 = 10\mathbf{I}_m$, $\mathbf{S}_1 = 2\mathbf{I}_n$ and $\mathbf{R}_1 = 2\mathbf{I}_n$. To estimate the state action value function using least squares we collect data for an episode of length $T = 500$ and thereafter update the policy based on the current estimate of $\mathrm{Q}_{\pi}$. Details of what the corresponding $A_{\psi}$, $A_{\nu}$, $B_x$ and $C$ matrices are for this problem are outlined in the Appendix of the paper.

\subsubsection{Inverted pendulum on a cart}

In the second example, we examine the classical problem of stabilizing an inverted pendulum. The system comprises of a pendulum of mass $\mathrm{m}$, length $\mathrm{l}$ and moment of inertia $I$ mounted on a mobile cart of mass $\mathrm{M}$ (see Figure \ref{inverted_pendulum}A). 
We linearized dynamics of this system under small angle approximations, resulting in:
\begin{equation}
	\label{pendulum_1}
	p_{t+1} = p_t +\Delta t v_t
\end{equation}
\begin{equation}
	\label{pendulum_2}
	\theta_{t+1} = \theta_t + \Delta t \omega_t
\end{equation}
\begin{equation}
	\label{pendulum_3}
	\begin{split}
		v_{t+1} = v_t + \Delta t( \frac{-(I+\mathrm{ml}^2)bv_t+\mathrm{m}^2g\mathrm{l}^2 \theta_t}{I(\mathrm{M}+\mathrm{m})+\mathrm{Mml}^2} +\\ \frac{I+\mathrm{ml}^2}{I(\mathrm{M}+\mathrm{m})+\mathrm{Mml}^2}\mathbf{b}\mathbf{x}_t)
	\end{split}
\end{equation}
and, 
\begin{equation}
	\label{pendulum_4}
	\begin{split}
		\omega_{t+1} = \omega_t + \Delta t( \frac{-(\mathrm{ml})bv_t+\mathrm{mgl(M+m)} \theta_t}{I(\mathrm{M}+\mathrm{m})+\mathrm{Mml}^2} +\\
		 \frac{\mathrm{ml}}{I(\mathrm{M}+\mathrm{m})+\mathrm{Mml}^2}\mathbf{b} \mathbf{x}_t)
	\end{split}
\end{equation}
\begin{figure}[b]
	\centering
	\includegraphics[width = \linewidth]{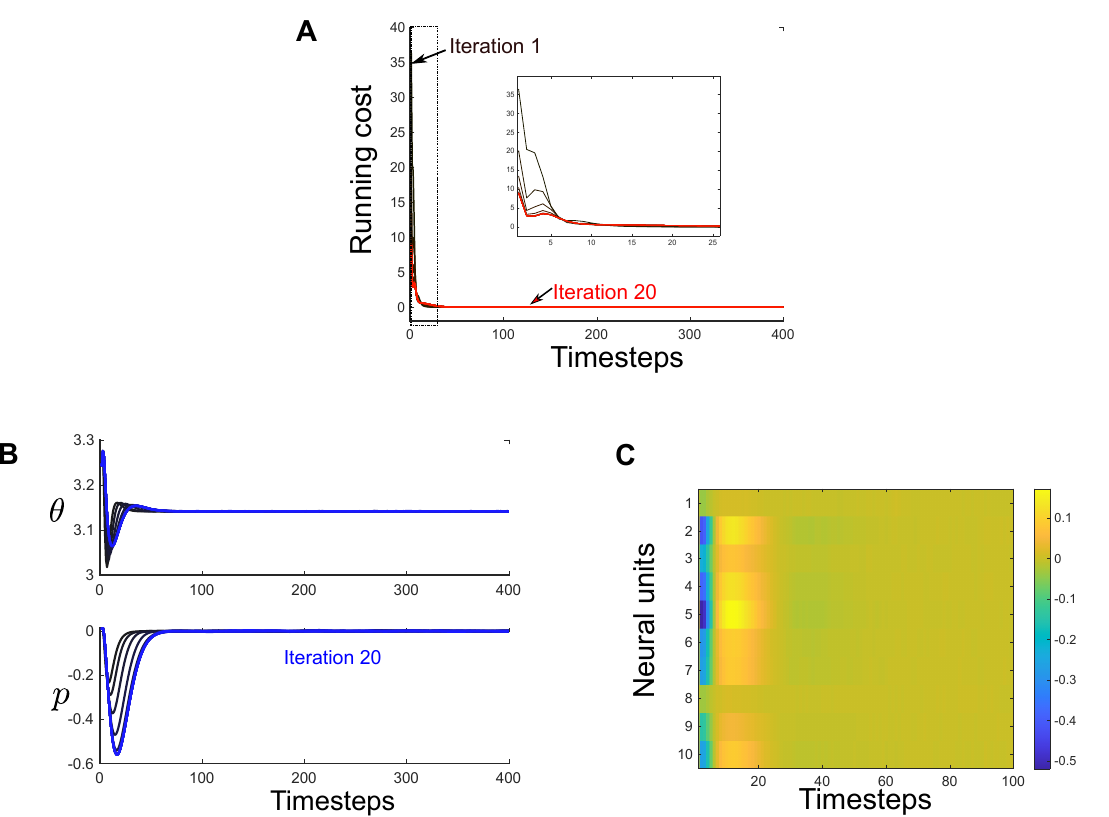}
	\caption{\textbf{Pendulum on a cart system} A-B. The network through episodic adaptation quickly learns the optimal strategy to perform the task. (The running cost in the first 25 timesteps of each iteration is shown inset in A). C. Activity of the network under the learned optimal strategy. }
	\label{inverted_pendulum}
\end{figure}

Here the variables $p_t$, $v_t$, $\theta_t$, $\omega_t$ correspond to position of the cart, velocity of the cart, angle from the vertically upright position and angular velocity respectively. $b$ here corresponds the coefficient of friction for the cart. The goal here is to stabilize the pendulum in the unstable equilibrium pointing up with minimum displacement of the cart (see Equation \eqref{pendulum_cost}). We find once again that the network adapts to perform the task optimally provided that the initial position of the pendulum is within 10 degrees of the vertical position (see Figure \ref{inverted_pendulum}). 
\begin{equation}
	\label{pendulum_cost}
	\mathbb{J} = \frac{1}{2}\sum_{t = 0}^{\infty}[\rho_pp_t^2+\rho_vv_t^2+\rho_{\theta}\theta_t^2+\rho_{\omega}\omega_t^2 + \mathbf{x}_t^T\mathbf{S}_2\mathbf{x}_t + \mathbf{u}_t^T\mathbf{R}_2\mathbf{u}_t]dt
\end{equation}
For the simulations here we have chosen the following parameter values: $\mathrm{m}=0.2$, $\mathrm{l}= 0.3$, $I = 0.006$, $\mathrm{M} = 0.5$, $b = 0.1$, $g = 9.8$ and $\Delta t = 0.1$. The weights of the matrix $\mathbf{b} \in \mathbb{R}^{1 \times n}$ are once again chosen from a uniform random distribution. Additionally, we have $\sigma_y^2 = 0.01$, $\gamma = 0.99$, $\rho_p = 1$, $\rho_v=1$, $\rho_{\theta} = 10$, $\rho_{\omega}=10$, $\mathbf{S}_2 = 2\mathbf{I}_n$ and $\mathbf{R}_2 = 2\mathbf{I}_n$. Here, we consider episodes of length $T = 400$ timesteps. We show in the Appendix what the matrices $A_{\Psi}$, $A_{\nu}$, $B_x$ and $C$ correspond to for this example.

An emergent trend observed through these simulations is the existence of sparsity in both architecture and dynamics of the network. It must be noted that through the objective function we promote high fidelity control in an energy efficient manner. We do not explicitly have a sparsity promoting regularizer term either on the activity $\mathbf{x}$ or on the optimal feedback matrix $\mathbf{W}_k$ in our proposed objective function. The fact that the network achieves so both in its dynamics and in its architecture is an interesting and unexpected property.

\subsection{Robustness benefits of distributing a policy in a network}

One of the key characteristics of distributed computation such as the one demonstrated here is its robustness, i.e., its reliable performance when there is degradation of activity of a subset of units. An example of a biological system that utilizes such distributed computation is the brain \cite{betzel2017multi}. It is well known that degradation of neurons and synapses routinely occur in the brain, and often these occur without any manifestation of neurological disorders. In this section, we investigate whether such properties are imbibed in the network that we synthesize through model-free techniques. 

To investigate this, we lesion a fraction of the network at three different instances of the policy iteration algorithm - (a) before commencement of the policy iteration (b) after the first $k_1$ rounds of policy iteration, and, (c) at the conclusion of policy iteration when an optimal strategy has been learned (see Figure \ref{robustness_networked_policy} A, B). Any lesion performed persists to the end of the simulation. Introducing this lesion in the network is carried out mathematically as follows. Recall, that the matrix $\mathbf{b}$ linearly combined contributions of individual units to formulate the control signal. We posit that when a lesion occurs, the network controls the physical system via $\mathbf{b} = [\mathbf{b'}: \mathbf{0}_{m \times (n-n_f)}]$ where $\mathbf{b'} = \mathbf{b}[1:m, 1:n_f]$ and $n_f$ is the number of units that are functioning, i.e., contribution of $(n-n_f)$ units is reduced to zero. 

\begin{figure}[t]
	\centering
	\includegraphics[width = \linewidth]{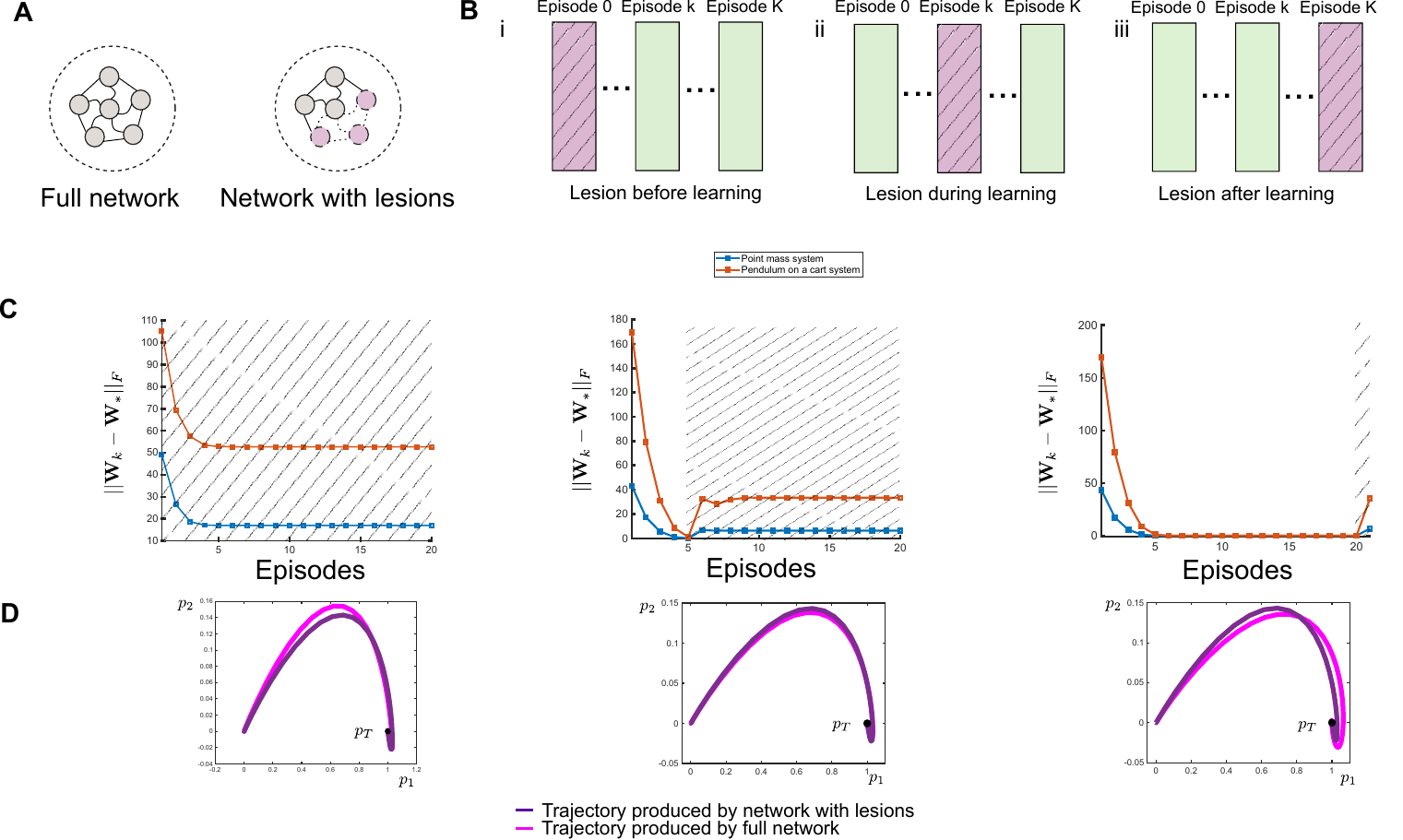}
	\caption{A. In a lesioned network contribution of a subset of units is removed. B. The network is disrupted via lesioning at three different times -(i) at the beginning of learning (ii) during learning (eg. at iteration $k_1 = 5$) and (iii) after the network has learnt an optimal strategy. The pink box with hatched pattern indicates when the disruption was introduced. C. Absolute difference of the strategy executed by the lesioned network and the full network (hatch pattern indicated the window for which the network was disrupted). D. Illustrative example of functioning of the network (i.e., navigation to a target location in a two-dimensional plane) with lesions in comparison to the full network for the point mass system.}
	\label{robustness_networked_policy}
\end{figure}

We observe that for both the point mass system as well as the inverted pendulum system in all three scenarios, the network is able to recover from the disruptions caused by lesion of the network and proceed with performing the task. Notably, when the lesion occurs before learning has begun, the network learns a policy that operates entirely without taking into consideration the units that were removed. As a result, the absolute difference between the learned policy and the true optimal policy constructed by the full network remains large, even though the task is performed with high fidelity (see Figure \ref{robustness_networked_policy} C, D). On the contrary, when the disruption occurs during the learning process or at the conclusion of the learning process, the network promptly compensates for it and proceeds to perform the task accurately albeit via a slightly sub-optimal policy. 

In this context it is worthwhile to mention that the performance of the network was much more robust to perturbations when it actuated the point mass system than when it stabilized the pendulum system (see Table \ref{robustness}). Intuitively, we can attribute this disparity in the algorithmic performance to the complexity of the system that is being controlled. The linear dynamical equations governing the pendulum on a cart given by \eqref{pendulum_1}, \eqref{pendulum_2}, \eqref{pendulum_3} and \eqref{pendulum_4} are approximations of the complex nonlinear dynamics around the unstable equilibrium point. Introducing perturbations in this system can therefore deflect it to regions in the state space where the linearized dynamics capture poorly the evolution of the system and therefore the estimates of the state action value function are inaccurate, which in turn causes policy iteration to diverge. In Table \ref{robustness}, we report success or failure of the network when a lesion was inflicted during the learning process. 

\begin{table}[t]
	\centering
	\includegraphics[width=\linewidth]{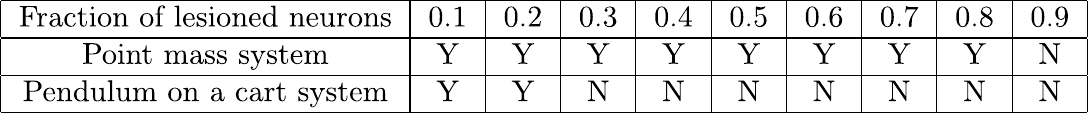}
	\caption[]{Robustness of network in the event of lesion}
	\text{Y: Network is able to recover and perform the task}
    \text{N: Network fails to perform the task}
   \label{robustness}
\end{table}

\subsection{Effects of tuning hyperparameters associated with policy iteration}

\subsubsection{Choice of episode length ($T$)}
One of the predominant challenges of designing control through simulations when the dynamics of the environment is unknown is ascertaining the associated sample complexity \cite{kakade2003sample}, i.e., determining how much experience must be simulated by the model for each round of policy iteration. When the state and action space are finite, then analytical bounds can be established over the number of samples with respect to the size of the state space, action space and the horizon for which the performance is desired \cite{kearns1999finite, kakade2003sample}. In this work, we consider however continuous and infinite state and action spaces and proceed by approximating the state action value function. In order to approximate the state-action value function $\mathrm{Q}_{\pi}(\bm{\Omega}_t, \mathbf{u}_t)$, we need to solve the least squares problem given by \eqref{error_min}. The number of parameters to be estimated in this problem is $(n'+n)^2$. However, if we consider symmetry of the $\mathbf{H}_{\pi}$ matrix, we need only estimate $\frac{1}{2}(n' + n + 1)(n' + n)$ parameters \cite{bradtke1994adaptive, lewis2012reinforcement}. As, in our formulation, $n' = n+2*m$ and $m<n$, without any loss of generality, we can say that in order to reliable estimates of $\mathrm{Q}_{\pi}$, the length of episodes needed scales as second order polynomial of the network size $n$. In Figure \ref{episode_length}, we have shown the quality of the learned policy as a function of length of episodes for a network of size $n=10$.
\begin{figure}[b]
	\centering
	\includegraphics[width= 0.9\linewidth]{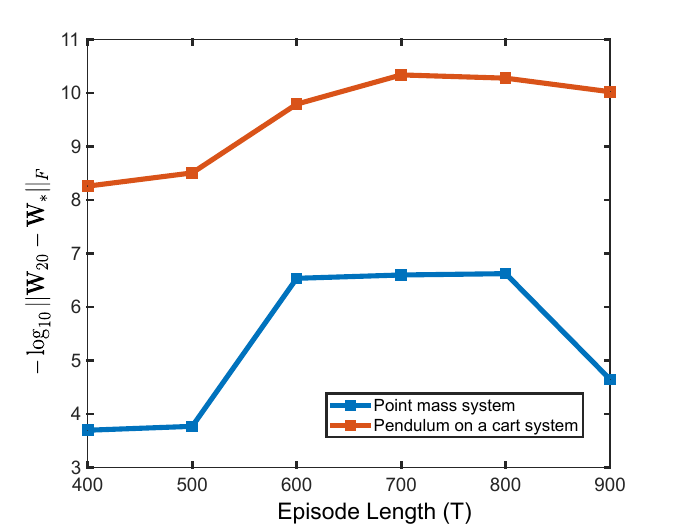}
	\caption{Quality of the solution provided after convergence of policy iteration as a function of length of episodes.}
	\label{episode_length}
\end{figure}

\subsubsection{Choice of discount factor $\gamma$}

For infinite horizon problems, to ensure that the sum of rewards converges i.e., the optimization problem is well defined, a discount factor $0 < \gamma < 1$ \cite{sutton2018reinforcement} is used. Intuitively, the discount factor acts as a parametric representation of `urgency'. Values of $\gamma$ closer to $0$ causes the network to care more for immediate costs and therefore be myopic, while that closer to $1$ causes the network to care more for future costs. We provide in Figure \ref{discount}, cost accrued over time steps for policies computed with different values for $\gamma$. 

\begin{figure}[t]
	\centering
	\includegraphics[width= \linewidth]{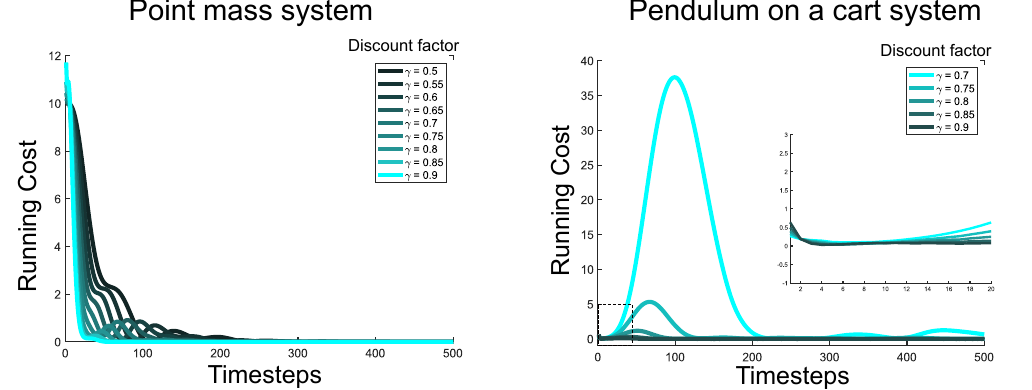}
	\caption{Running cost over timesteps across different discount factors. For smaller values of $\gamma$, the optimal strategy focuses on reducing costs earlier in the episode rather than later. (Inset) Costs accrued by the pendulum system in the first 25 timesteps.}
	\label{discount}
\end{figure}

\subsubsection{Exploration vs exploitation}

As mentioned in the Problem Formulation section, we take a different route than predominantly used $\epsilon$-greedy policy for exploration in this paper. We introduce a systematic exploration technique, i.e.,  at any given iteration, data over an episode is collected by perturbing the system via the input: $\mathbf{u}_t = \pi_k(\bm{\Omega}_t) + \varepsilon_t$, where $\varepsilon_t \sim \mathcal{N}(0, \sigma_y^2 \mathbf{I}_n)$. This conservative form of exploration ensures that other actions than given by the policy update step are investigated while maintaining that the policy chosen will mostly be stabilizing i.e., it would not cause the system to explode as it is simulated. When the system explodes owing to a randomly chosen policy that is not stabilizing, the bounds on performance as established in the convergence analysis section of this performance do not hold true and the algorithm fails to converge. 

In Figure \ref{exploration_figure}, we show for both the point mass system and the pendulum on a cart system, range of policies selected to probe the system under different degrees of exploration. Note that here we coin $\sigma_y^2$ as the degree of exploration. Intuitively, when $\sigma_y^2$ assumes a higher value, the system is encouraged to try wider range of action policies around the policy prescribed by the update step of the algorithm. 

\begin{figure}[b]
	\centering
	\includegraphics[width= \linewidth]{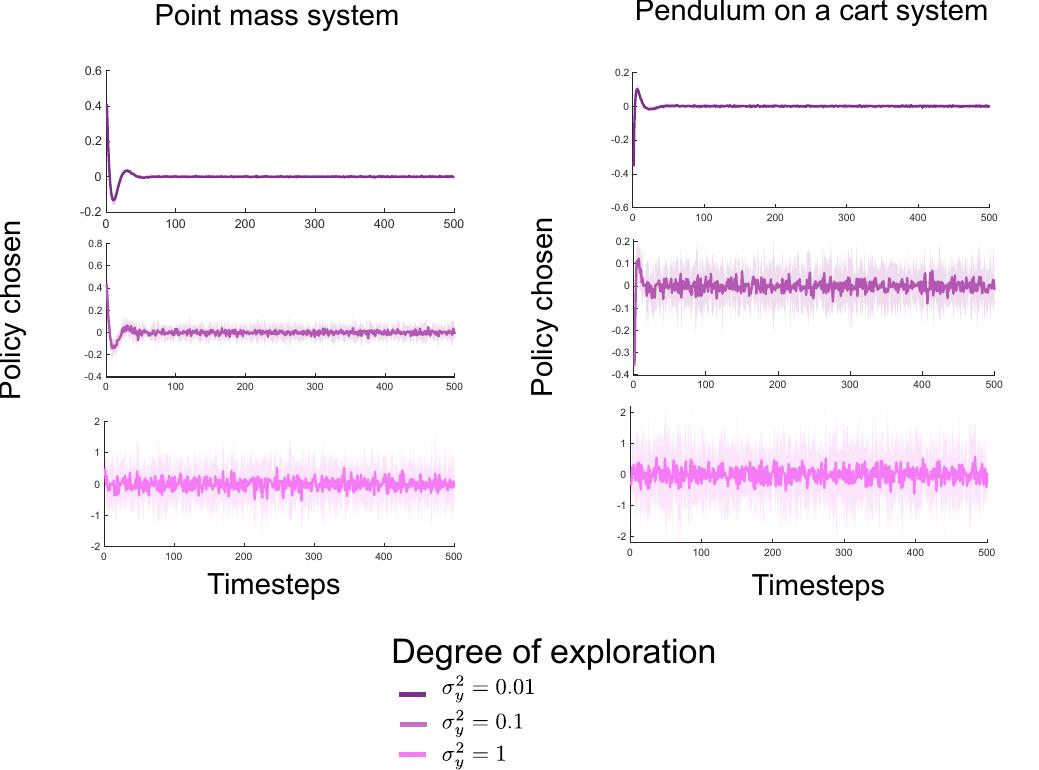}
	\caption{Policies (mean $\pm$ 2std) over time steps used to generate simulation data for different degrees of exploration. When higher degrees of exploration are considered, variance of policies considered is significantly higher. }
	\label{exploration_figure}
\end{figure}

\section{Discussion}

\subsection{Summary} 

In this work, we have provided a framework for synthesizing a distributed, network-based controller that can be adapted in order to manipulate linear dynamical systems. The networks are built by using an augmented state space, facilitating direct synthesis of network dynamics by means of solving a control objective. This method is analytically amenable to least-squares parametric adaptation, thus yielding the overall scheme for distributed control of unknown systems.

Our framework uses an online least squares approximate policy iteration method to adapt the controller (see Algorithm \ref{algo}). This algorithm has two key steps: evaluating the efficacy of the realized policy by means of a state action value function (see Problem Formulation), and then updating the policy based on this evaluation. The algorithm repeats these two steps until a stopping condition is reached. As such, the overall controller can be thought of as a network with two time-scales. Through the outer time-scale, the feedback matrix $\mathbf{W}_k$ is updated epsiodically. On the other hand, the inner time-scale is associated with the dynamics of the network itself (see Equation \eqref{inner_time_scale_2}), dictating how it generates activity and ultimately, the signals that will drive the plant. 

\subsection{Tractability}

An advantage of our framework is tractability for analysis of when we can expect this model-free policy iteration to reach a `good' solution. One of the challenges with convergence analysis of online least squares based methods have been the fact that the policy update step inherently depends on the estimate of state action value function in the policy evaluation step \cite{bertsekas1996neuro, bucsoniu2012least, friedrich2019least}. However, because we perform updates on a prior architecture that is analytically associated with a well-defined optimal control problem, we can probe conditions and bounds for convergence to a true optimal solution. More precisely, if we start with a stabilizing policy $\pi$ and gather enough data (scales linearly with the number of parameters to be estimated) so that the least squares problem is well posed, then we obtain a stabilizing policy rapidly (in the next update). This sequence of stabilizing policies will converge and approach the optimal policy. 

\subsection{Robustness of the distributed controller}

One of the key premises for distributing a controller onto a network is robustness \cite{gupta2006robustness}. In other words, when a subsection of the network fails to contribute to the task, the remaining units can compensate for it and ensure that the task is completed. To demonstrate that the network we synthesized in this work embeds this robustness property, we manually removed certain units during different phases of the iterative procedure. Once removed, these units were not added back to the network. Depending on when this `lesioning' procedure was carried out, the network adapted differently to complete the task. We found that the network was particularly robust for controlling the point mass system, wherein removal of as much as $80\%$ of the network allowed task completion albeit with slightly sub-optimal strategies. Because the pendulum model is a linearization, perturbation of the network is potentially less robust as it may result in a policy wherein the state departs from a neighborhood of the equilibrium in question.  In general, the precise ratio of units that can be lesioned is likely a function of the complexity of the plant dynamics, relative to the number of units in the controller network.  

Robustness of the approach extends to choice of hyperparameters such as episode length $T$, discount factor $\gamma$ and degree of exploration $\sigma_y^2$. We find that episode length scales as a second order polynomial function of the network size $n$. This is expected as in the policy evaluation step, we must estimate the state action value function which is parametrically represented, the number of parameters increasing with size of the network that is used to generate control. We also provide analysis of observations noted by varying how much discounted future costs were and how much exploration was executed by the network. 

\subsection{Features not explained}

There are a number of important caveats and limitations that must be pointed out regarding this work. Most notably, we have limited our derivation at this point to linear systems, though our recent work \cite{mallik2021topdown} provides a basis for potential future extension to certain nonlinear systems as well. 


In this work, we have identified the two timescales emergent from the algorithm . We have provided a closed form for the network activity  (see Equation \eqref{inner_time_scale_2}) which receives feedback from the environment. Comparing this distributed solution scheme to the activity of a network of \textit{firing rate} neurons and synapses is challenging. \cite{dayan2001theoretical}. This is because the adaptive dynamics shown in Equation \eqref{outer_time_scale} are not biological in nature, since they rely on solving a global optimization problem. One possible extension to reconcile this issue is to use gradient based methods in the policy evaluation step \cite{lewis2012reinforcement}.
For example, 
\begin{equation}
	\bm{\theta}_{k}^{i+1} = \bm{\theta}_{k}^{i} - \eta \Phi_t((\bm{\theta}_{k}^{i})^T\Phi_t - \mathbf{c}_t)
\end{equation}
These methods are not without their limitations particularly when dealing with discrete time continuous state space problems. For instance, finding an initial choice of parameters $\bm{\theta}^0$ for which the resulting feedback matrix is stabilizing is non-trivial \cite{fazel2018global, park2020structured}. Secondly, the choice of step size $\eta$ in these gradient based approaches significantly impacts the convergence of the algorithm \cite{fazel2018global, schulman2015trust, schulman2017proximal, park2020structured} particularly when updating based on a single observation $(\bm{\Omega}_t, \mathbf{c}_t, \bm{\Omega}_{t+1})$. Gradient smoothing by simulating multiple trajectories and using a small fixed horizon for collecting samples instead of a single time point introduces some tractability in terms of convergence \cite{fazel2018global, park2020structured} to a solution. However, under these steps the convergence becomes sensitive to the collective choice of hyperparameters such as number of trajectories simulated, horizon selected etc. Finally, heuristic methods such as line search can also be used \cite{park2020structured} to improve algorithmic performance, but reconciling how these heuristics translate to biologically plausible computation remains to be addressed.

\section{Appendix}

\subsection{Reformulation of optimization problems to the discrete time LQR format}

In this section we provide the explicit forms for the matrices $\mathbf{A}$ and $\mathbf{B}$ used for specifying the dynamics of the system as well the penalty matrices $\mathbf{Q}$, $\mathbf{R}$ used to compute the cost at each timestep.

\subsubsection{Point mass system}

In this problem, we have $\bm{\Psi}_t \equiv \mathbf{p}_t$ and $\bm{\nu}_t \equiv \mathbf{v}_t$. The dynamics is governed by $C = \Delta t \mathbf{I}_m$, $A_{\Psi} = \mathbf{0}_{m\times m}$, $A_{\nu} =(1 - \Delta t\lambda_v) \mathbf{I}_m$ and $B_x =\Delta t \mathbf{b}$. The penalty matrices are given by $\mathbf{Q} = \left[\begin{array}{ccc}
	\mathbf{Q}_1 & & \\
	 &\mathbf{0}_{m \times m} &\\
	 & &\mathbf{S}_1
\end{array}\right]$ and $\mathbf{R} = \mathbf{R}_1$. 

\vspace{0.5cm}

\subsubsection{Pendulum on a cart system}

In this problem, we have $\bm{\Psi}_t \equiv [p_t, \theta_t]^T$ and $\bm{\nu}_t \equiv [v_t, \omega_t]^T$. The dynamics is governed by $C = \Delta t\mathbf{I}_m$, $A_{\Psi} = \Delta t\left[\begin{array}{cc}
	0 &\frac{\mathrm{m}^2g\mathrm{l}^2}{I(\mathrm{M}+\mathrm{m})+\mathrm{Mml}^2} \\
	0 &\frac{\mathrm{mgl(M+m)}}{I(\mathrm{M}+\mathrm{m})+\mathrm{Mml}^2}
\end{array}\right]$, $A_{\nu} = \mathbf{I}_m + \Delta t\left[\begin{array}{cc}
\frac{-(I + \mathrm{ml}^2)b}{I(\mathrm{M}+\mathrm{m})+\mathrm{Mml}^2} &0\\
 \frac{\mathrm{m^2gl^2}}{I(\mathrm{M}+\mathrm{m})+\mathrm{Mml}^2} &0
\end{array}\right]$, and, $B_x =\Delta t \left[\begin{array}{c}
\frac{(I+\mathrm{ml}^2)\mathbf{b}}{I(\mathrm{M}+\mathrm{m})+\mathrm{Mml}^2} \\
\frac{\mathrm{ml}\mathbf{b}}{I(\mathrm{M}+\mathrm{m})+\mathrm{Mml}^2}
\end{array}\right]$. The penalty matrices are given as $\mathbf{Q} = \left[\begin{array}{ccccc}
\rho_p & & & &\\
 &\rho_{\theta} & & &\\
 & &\rho_v & &\\
 & & &\rho_{\omega} &\\
 & & & &\mathbf{S}_2
\end{array}\right]$ and $\mathbf{R} = \mathbf{R}_2$. These linearizations hold within $\pm 10$ degrees of the unstable fixed point.

\subsection{Proof of Lemma \ref{lemma_3}}

We know, $||\hat{\mathrm{Q}} - \mathrm{Q}_*|| \leq \epsilon$. Now, 

\begin{equation}
	\begin{aligned}
	\mathrm{Q}_*(\bm{\Omega}_t, \pi(\bm{\Omega}_t)) - \mathrm{V}_*(\bm{\Omega}_t) 
	&= \mathrm{Q}_*(\bm{\Omega}_t, \pi(\bm{\Omega}_t)) - \hat{\mathrm{Q}}(\bm{\Omega}_t, \pi(\bm{\Omega}_t))+\\ &\hat{Q}(\bm{\Omega}_t, \pi(\bm{\Omega}_t)) - \mathrm{V}_*(\bm{\Omega}_t) \\
	&\leq \hat{Q}(\bm{\Omega}_t, \pi(\bm{\Omega}_t)) - \mathrm{V}_*(\bm{\Omega}_t) + \epsilon \\
	&= \hat{\mathrm{Q}}(\bm{\Omega}_t, \pi(\bm{\Omega}_t)) - \mathrm{Q}_*(\bm{\Omega}_t, \pi_*(\bm{\Omega}_t)) + \epsilon \\
	&\leq 2\epsilon
	\end{aligned}
\end{equation}
The last step is possible as by construction of $\pi$, $\hat{\mathrm{Q}}(\bm{\Omega}_t, \pi(\bm{\Omega}_t)) \leq \hat{\mathrm{Q}}(\bm{\Omega}_t, \pi_*(\bm{\Omega}_t))$.

We can now write that:
\begin{equation}
	\begin{aligned}
		\mathrm{V}_{\pi}(\bm{\Omega}_t) - \mathrm{V}_*(\bm{\Omega}_t) 
		&= \mathrm{V}_{\pi} - \mathrm{Q}_*(\bm{\Omega}_t, \pi(\bm{\Omega}_t)) + \mathrm{Q}_*(\bm{\Omega}_t, \pi(\bm{\Omega}_t)) - \mathrm{V}_* \\
		&\leq \mathrm{V}_{\pi} - \mathrm{Q}_*(\bm{\Omega}_t, \pi(\bm{\Omega}_t)) + 2\epsilon \\
		&= \mathrm{Q}(\bm{\Omega}_t, \pi(\bm{\Omega}_t)) - \mathrm{Q}_*(\bm{\Omega}_t, \pi(\bm{\Omega}_t)) + 2\epsilon \\
		&= \mathrm{Q}(\bm{\Omega}_t, \pi(\bm{\Omega}_t)) - \mathbf{c}_t + \mathbf{c}_t - \mathrm{Q}_*(\bm{\Omega}_t, \pi(\bm{\Omega}_t)) + 2\epsilon \\
		&= \gamma \left[\mathrm{V}_{\pi}(\bm{\Omega}_{t+1}) - \mathrm{V}_{*}(\bm{\Omega}_{t+1})\right] + 2\epsilon
	\end{aligned}
\end{equation}
By recursing on this equation, we have $\mathrm{V}_{\pi}(\bm{\Omega}_t) - \mathrm{V}_*(\bm{\Omega}_t) \leq 2\epsilon (1 + \gamma + \gamma^2 +...) = \frac{2\epsilon}{1 - \gamma}$. Taking $\epsilon' = \frac{2\epsilon}{1 - \gamma}$ norm thereafter proves the Lemma.

\section*{Acknowledgment}

This work has been supported by Grants CMMI - 1653589 and EF-1724218 from the National Science Foundation. 

\bibliographystyle{IEEEtran}
\bibliography{references}

\end{document}